\def\beq{\begin{equation}}
\def\eeq{\end{equation}}
\def\beqa{\begin{eqnarray}}
\def\eeqa{\end{eqnarray}}
\def\bmat{\begin{matrix}}
\def\emat{\end{matrix}}
\def\bmli{\begin{multline}}
\def\emli{\end{multline}}
\def\C{{\mathcal C}}

\def\Ne{\mathbb{N}}

\def\R{\mathbb{R}}
\def\Pe{\mathbb{P}}

\def\al{\alpha}
\def\be{\beta}
\def\ga{\gamma}

\def\del{\delta}

\def\la{\lambda}
\def\La{\Lambda}
\def\om{\omega}
\def\Om{\Omega}

\def\sig{\sigma}
\def\vphi{\varphi}
\def\p{\partial} 
\def\ti{\tilde} 
\def\nea{\nearrow}

\def\bra{\langle}
\def\ket{\rangle}
\def\12{{\textstyle{1\over2}}}
\def\we{\wedge}

\documentclass{amsart}
\usepackage{amsmath,amsthm,amssymb,latexsym,xcolor}
\usepackage[notref,notcite]{showkeys}
\usepackage{epsfig}
\usepackage{enumerate}
\newtheorem{lemma}{Lemma}
\newtheorem{theorem}{Theorem}
\newtheorem{proposition}{Proposition}
\newtheorem{remark}{Remark}
\begin{document}
\title{Galton-Watson trees with first ancestor interaction}
\author[F. Dunlop]{Fran\c cois Dunlop}
\address{Laboratoire de Physique Th{\'e}orique et Modélisation (CNRS,
  UMR 8089)\\ 
CY Cergy Paris Universit{\'e}, 95302 Cergy-Pontoise\\
France}
\email{francois.dunlop@cyu.fr}
\author[A. Mardin]{Arif Mardin}
\address{
Nesin Matematik Köyü,
Şirince mahallesi, 7, Kayserkaya Sokak,
35920 Selçuk, İZMİR,
Turkey.}
\email{mardin.arif@gmail.com}
\begin{abstract}
We consider the set of random Bienaymé-Galton-Watson trees with a bounded number
of offspring and bounded number of generations as a statistical mechanics model:
a random tree is a rooted subtree of the maximal tree; the spin at a given node
of the maximal tree is equal to the number of offspring if the node is present
in the random tree and equal to -1 otherwise. We introduce nearest neighbour
interactions favouring pairs of neighbours which both have a relatively large
offspring. We then prove (1) correlation inequalities and (2) recursion
relations for generating functions, mean number of {    external nodes}, interaction energy
and the corresponding variances. The resulting quadratic dynamical system, in
two dimensions or more depending on the desired number of moments, yields almost
exact numerical results. The balance between offspring distribution and coupling
constant leads to a phase diagram for the analogue of the extinction
probability. On the transition line the mean number of {    external nodes} in generation $n{    +1}$
is found numerically to scale as $n^{-2}$.
\end{abstract}
\date{\today}
\maketitle
\section{Introduction}
The principal tool of our investigation is a Bienaymé-Galton-Watson (BGW) tree, which is a specific example of those stochastic processes known as “branching processes”.  As such, it has found many fruitful applications not only in population dynamics, but also in genetics, nuclear chain reactions, etc. It has an interesting history regarding its origins, which we shall try to outline below.
In most textbooks treating the subject, the tree in question is usually called “Galton-Watson tree” because the initial impetus is attributed to the question posed by the British statistician Francis Galton in “Educational Times” in 1873 \cite{G1873}, concerning the possibility of extinction of the names of “noble” families in Britain.  After having received several incorrect solutions, he managed to rouse the interest of his mathematician friend Henry William Watson, who posed the problem correctly but his solution was not quite right: He concluded that the probability of extinction is always (i.e., including the supercritical case) equal to 1 \cite{W1873,GW1874}. At that time the matter was believed to be settled by the solution of Watson.

Publication of a completely correct solution had to wait until 1930 \cite{S30}, when the Danish mathematician J.F. Steffensen’s work using contemporary probabilistic tools was published, in a danish mathematical journal. His article being in danish must surely have prevented its immediate recognition. Nevertheless, three years later he published a more comprehensive version of his work in Annales de l’Institut Henri Poincaré \cite{S33}, written in French, thus enabling him to reach a much wider part of the global mathematical community. 

Interestingly enough, this is not the whole story as far as historical precedence is concerned. The names of A.K. Erlang, J.B.S. Haldane and, above all, Jules Bienaymé (the same Bienaymé of the famous Bienaymé-Chebychev Inequality of probability theory) should also be mentioned. We shall henceforth focus on Bienaymé’s contribution, and invite the reader to consult David Kendall’s article \cite{K66} for a nice exposition of the work of Erlang and Haldane.

Surprisingly, Jules Bienaymé considered the same kind of problem regarding the French aristocracy and famous bourgeois families in France nearly three decades before the appearance of Galton’s famous problem in Educational Times. His communication was published in the journal of the Société Philomatique de Paris in 1845 \cite{B1845}. This important discovery was made in 1972 by C.C. Heyde and E. Seneta \cite{HS72}. What is striking is that although Watson’s solution led to the erroneous conclusion that there would be extinction with probability one even in the supercritical case, Bienaymé had the whole theorem of criticality correctly posed. However, even though a full treatment of the problem for publication has been promised by him in a “mémoire spécial”, no trace of it has been found so far. It seems that nobody can be sure if there is yet another and yet even earlier serious attempt to state and prove the same result waiting to be discovered. In our treatment we shall honour the historical precedence and call our tree a Bienaymé-Galton-Watson  (BGW) tree.
Finally, interested readers are highly recommended to consult the two excellent survey articles by David Kendall \cite{K66,K75}.

We are interested in large random planar (ordered, labeled) rooted trees, which 
may have grown by a Markov process such as a Galton Watson process, and are then
 subject to self-interaction through a Boltzmann weight, with an interaction
energy between first ancestor and offspring.
{    The main motivation behind this work is an attempt to understand the following problem. Being given a BGW tree in the usual sense, to what degree can one change the behaviour of the characteristic values such as the mean number of offspring or extinction probabilities by attributing some additional probabilities to those existing between typical parent-offspring pairs: More precisely, what happens if the offspring of a populous family (i.e. parents having at least two children) have populous families themselves? Such relations being nearest-neighbour type due to the tree structure, we introduced an interaction function which increases the likelihood of this sort of outcome. Lattice models of statistical mechanics employ similar ideas, and we tried to adapt its approach to our problem.
Somewhat similar questions have been treated in the search of the (discrete) time necessary to find the most recent common ancestor of a large population by M.Möhle \cite{M04}.

Two notable features come with our approach. On one side the random tree hierarchical structure yields a fruitful reduction to a dynamical system. On another side the extinction problem of the BGW Markov chain, transferred to the statistical mechanics framework, bears some similarity with the pinning/depinning transition: average tree height bounded or going to infinity in the infinite volume limit.
Whence some analogy with the work of Derrida and Retaux \cite{dr14} about the depinning transition
with disorder on the hierarchical lattice, a toy model of which they solve
using the quadratic map of Collet-Eckmann-Glaser-Martin \cite{cegm84}.

The outline of the paper is as follows.}
In Section~\ref{model} we define the model precisely in terms of a Gibbs measure.
In Section~\ref{mc} we design a Markov chain under which the Gibbs measure is
invariant.
In Section~\ref{spinmodel} we convert random trees to random spin configurations
on the maximal tree, where the spin value at any node is -1 if the branch is
extinct and the number of offspring otherwise.
In Section~\ref{correlation} we prove correlation inequalities of the
Griffiths and FKG types.
In Section~\ref{recursion} we establish a recursion relation for generating
functions of two parameters, a dynamical system where ``time" counts the number
of generations and two-dimensional space $\R^2$ represents the two activities
associated with single offspring and two or more offspring respectively. 
In Section~\ref{leaves} we extend the dynamical system to $\R^6$ to obtain a
recursion for the mean number of {    external nodes}.
In Section~\ref{energy} we extend the dynamical system to $\R^4$ to obtain a
recursion for the mean energy and to $\R^6$ to obtain a
recursion for the variance of the energy.
In Section~\ref{fixed}, in the case where the number of offspring is 0, 1, or 2,
we find
conditions for the presence of a fixed point in the relevant physical domain.
This fixed point is expected to correspond to subcritical or critical trees.
In the critical case, running numerically the dynamical system, we find that
the mean number of {    external nodes} scales like $n^{-2}$ as $n\nea\infty$, contrasting
with the constant value 1 in the noninteracting critical BGW case.
In Section~\ref{p_1p_2}, in the case where the number of offspring is 1 or 2,
we prove convergence of the specific free energy as $n\nea\infty$, and study
numerically its dependence upon the coupling constant.

\section{Model}\label{model}
For any node $i$ in a tree $\om$, let $X_i\in\Ne=\{0,1,2,\dots\}$ denote the 
number of offspring of $i$. The root is node 0. Nodes are labeled \`a la Neveu 
\cite{N86}:
\beq
X_0,X_1,\dots X_{X_0},X_{11},\dots X_{1X_1},X_{21},\dots X_{2X_2},\dots 
\eeq
The generation of a node is its distance to the root. Except for the root, which
belongs to generation 0, the generation of a node equals its number of digits in
Neveu notation. It will be denoted $|i|$. The number of generations of a finite
tree $\om$ is
\beq
|\om|=\max\{|i|:\ i\in\om\}
\eeq 
We denote $a(i)$ the parent (first ancestor) of $i$, and $r(i)$ the rank of 
$i$ within its family, the last digit of its Neveu label. 

A probability measure on the set of planar rooted trees $\om$ with no more than 
$n$ generations is defined as follows. For definiteness we take  as 
{\sl ``non-interacting''} reference measure a BGW probability 
distribution
\beq
\Pe^{GW}(\om)=\prod_{i\in\om}p_{X_i}
\eeq
with $p_k=\Pe^{GW}(X_0=k)$ of bounded support: $2\le K=\max\{k:p_k>0\}<\infty$.
{    The case of unbounded support will be considered in a forthcoming paper
\cite{CDHM22}.}
We denote $\bar k=\sum k\,p_k$.
Expectations in $\Pe^{GW}$ will be denoted $\bra\cdot\ket_{GW}$ and 
$\bra A;B\ket_{GW}=\bra AB\ket_{GW}-\bra A\ket_{GW}\bra B\ket_{GW}$. We recall
\beq
\Bigl\bra X_0\Bigr\ket_{GW}=\bar k,\qquad
\Bigl\bra X_1+\dots+X_{X_0}\Bigr\ket_{GW}=\bar k^2,\qquad
\Bigl\bra\sum_{|i|=m\atop i\in\om}X_i\Bigr\ket_{GW}=\bar k^{m+1}
\eeq
Then, given a pair interaction energy
\beq\label{phi1}
\{0,\dots,K\}\times\{0,\dots,K\}\,\ni(X,Y)\longmapsto\,\vphi(X,Y)\,\in\R
\eeq
and a boundary condition $X_{a(0)}=x\in\{1,\dots,K\}$ specifying the offspring of
a virtual ancestor for the origin, we define a Hamiltonian with first ancestor
interaction,
\beq\label{H}
H^x(\om)=\sum_{i\in\om} \vphi(X_{a(i)},X_i)
\eeq
and a probability measure on the set of trees $\om$ with at most $n$
generations,
\beq\label{Pn}
\Pe^x_n(\om)=\Bigl(\Xi^x_n\Bigr)^{-1}\Pe^{GW}(\om)
e^{-\be H^x(\om)}
\eeq 
\beq\label{Xin}
\Xi^x_n=\sum_{\om}\Pe^{GW}(\om)e^{-\be H^x(\om)}
\eeq
where $\be\ge0$ is the inverse temperature. Expectations in $\Pe^x_n$ will
be denoted $\bra\cdot\ket_n^x$ and
$\bra A;B\ket_n^x=
\bra AB\ket_n^x-\bra A\ket_n^x\bra B\ket_n^x$.
We shall be particularly interested in the average total offspring in
generation $m$, or average population in generation $m+1$,
in a tree with $n$ generations,
\beq
\Bigl\bra\sum_{|i|=m\atop i\in\om}X_i\Bigr\ket_n^x
\eeq
and in the average total energy
\beq
\Bigl\bra H^x \Bigr\ket_n^x
\eeq

The tree represents a hierarchical network, and the nodes are centers of 
activity. The activity of a node $i$ is measured by its offspring $X_i$,
which may be considered as the number of affiliated centers of activity. The 
activities of parent and child nodes are 
expected to be positively correlated, although this could depend upon the type 
of network. Our basic example will be
\beq\label{phi}
\vphi(X,Y)=-f(X)g(Y) 
\eeq
where $f$ and $g$ are non-negative non-decreasing functions on
$\{0,\dots,K\}$. For example
\beq\label{22}
\vphi(X,Y)=-1_{X\ge2}1_{Y\ge2} 
\eeq
where the indicator function $1_A$ takes value one if event $A$ is true and zero
otherwise. 

\section{Detailed balance with respect to $\Pe^x_n(\cdot)$}\label{mc}
In order to clarify a role that time can play in our study, we now define
two {    discrete-time} Markov chains obeying the detailed balance condition
with respect to $\Pe^x_n(\cdot)$. {    Both Markov chains to be defined below are irreducible and aperiodic in their respective state spaces.}
The first one is defined as follows. Draw the initial 
configuration with at most $n$ generations from  $\Pe^{GW}(\cdot)$. Then for
$t\ge0$ take transition probabilities
\beq
\Pe^x(\om^{t+1}=\om'|\,\om^t=\om)
=\Pe^{GW}(\om')e^{-\be (H^x(\om')-H^x(\om))_+}\,,\quad\om'\ne\om
\eeq
and $\om^{t+1}=\om$ with the {    complementary} probability.
The chain may be coupled to an i.i.d. sequence
$\bar\om^0,\bar\om^1,\dots,\bar\om^t,\bar\om^{t+1},\dots$ where each $\bar\om^t$ 
is drawn independently from $\Pe^{GW}(\cdot)$. Take the same initial condition
$\bar\om^0=\om^0$, and then inductively
\beq
\om^{t+1}=\eta_{t+1}\bar\om^{t+1}+(1-\eta_{t+1})\om^t
\eeq
where $\eta_{t+1}=1$ with probability $e^{-\be (H^x(\om^{t+1})-H^x(\om^t))_+}$
and $\eta_{t+1}=0$ with the {    complementary} probability. Detailed balance with respect
to (\ref{Pn}) is clearly satisfied. Such a dynamics making huge (macroscopic)
steps is not very useful.

Our second Markov chain obeying the detailed balance
condition with respect to $\Pe^x_n(\cdot)$ is more like usual Monte Carlo
dynamics, and is defined as follows. Draw the initial configuration from
$\Pe^{GW}(\cdot)$. Then at each time step:
\begin{itemize}
\item 
Pick a generation $m\in\{0,1,\dots n\}$ randomly according to some
probability distribution $\{\la_m\}_{m\in\{0,\dots,n\}}$ such that $\la_0>0$.
\item 
Pick a rank $i_1,\dots,i_m$ à la Neveu, each $i_l$ independently with
probability $p_{i_l}$.
  \item 
If the corresponding site $i\notin\om$, do nothing and exit the time step.
  \item Flip a fair coin, $\sig=\pm1$ with equal probabilities.
  \item If $X_i+\sig\notin\{0,\dots,K\}$, do nothing and exit the time step.
\item If $\sig=1$, with probability $p_{X_i+1}$ let $X'_i=X_i+1$ and from the
additional node, on the right of the $X_i$ already existing nodes, draw a tree
with $n-|i|$ generations from $\Pe^{GW}_{n-|i|}(\cdot)$.
\item 
If $\sig=-1$, with probability $p_{X_i-1}$ let $X'_i=X_i-1$ by removing the
rightmost node stemming from $i$ and the associated sub-tree.
\item 
If a new $X'_i$ has been defined, accept the new configuration $\om'$ with
probability $e^{-\be (H^x(\om')-H^x(\om))_+}$. 
\end{itemize}
Note that $X_i'=X_i+1\Rightarrow H^x(\om')<H^x(\om)$.
Two configurations $\om$ and $\om'$ are connected by one step of the Markov
chain if and only if $\exists\,i\in\om,\, i\in\om'$ such that $X'_i-X_i=\pm1$
and $\forall\,j\ne i,\,j\in\om\cap\om'$, $X'_j=X_j$. One can check detailed
balance, with $X'_i=X_i+1$ for definiteness,
\beq
{\Pe(\om\to\om')\over\Pe(\om'\to\om)}=
{\Pe^x_n(\om')\over\Pe^x_n(\om)}=
{p_{X'_i}\over p_{X_i}}\Pe^{GW}_{n-|i|}(\om'\setminus\om)\,
e^{-\be \bigl(H^x(\om')-H^x(\om)\bigr)_+}
\eeq
where $\om'\setminus\om$ is the configuration $\om'$ restricted to nodes not in
$\om$.
The dynamics can be accelerated by dividing all
transition rates by $\max\{p_k\}_{0\le k\le K}$.

\section{Random tree as a spin model}\label{spinmodel}
For any label $i$, for any planar rooted tree $\om$, we define $ X_i(\om)$
as the number of offspring of $i$ if $i\in\om$, and -1 otherwise. Therefore
$X_{i_1\dots i_k}=0\ \Rightarrow\ X_{i_1\dots i_ki_{k+1}}=-1$ $\forall i_{k+1}$. More
generally $i_{k+1}>X_{i_1\dots i_k}\ \Rightarrow X_{i_1\dots i_ki_{k+1}}=-1$.
\begin{proposition}\label{spin}
Let $n\ge1$.
Let $\La_n$ be the {    maximal} tree obtained with $X_i=K\ \forall i$, rooted at 0, with
at most $n$ generations. The number of sites $i\in\La_n$ is
\beq\label{Lan}
|\La_n|=1+K+K^2+\dots+K^n={K^{n+1}-1\over K-1}
\eeq
Let
\beq
\Om_n=\{-1,0,1,\dots,K\}^{\La_n}
\eeq
By convention, let $p_{-1}=1$. Fix a boundary condition
$X_{a(0)}\in\{1,\dots,K\}$. For $\chi=\{X_i\}_{i\in\La_n}\ \in\ \Om_n$ let
\beq\label{muGW}
\mu^{GW}(\chi)=\prod_{i\in\La_n}p_{X_i}
\Bigl(1_{X_{a(i)}\ge r(i)}1_{X_i\ge0}+1_{X_{a(i)}<r(i)}1_{X_i<0}\Bigr)
\eeq
where the product counting measure on $\Om_n$ is understood. Recall
\beq
\Pe^{GW}(\om)=\prod_{i\in\om}p_{X_i}
\eeq
Then there is a bijection $\chi\leftrightarrow\om$ between the support of
$\mu^{GW}$ and the set of BGW trees with at most
$n$ generations, and $\mu^{GW}\sim \Pe^{GW}$:
\beq\label{omchi}
\chi\mapsto \om(\chi),\quad\Pe^{GW}(\om(\chi))=\mu^{GW}(\chi)\ ;\qquad
\om\mapsto\chi(\om),\quad\mu^{GW}(\chi(\om))=\Pe^{GW}(\om)
\eeq
Moreover, let
\beq\label{HGW}
H^{GW}(\chi)=-\sum_{i\in\La_n}\Bigl(1_{X_{a(i)}\ge r(i)}1_{X_i\ge0}+1_{X_{a(i)}<r(i)}1_{X_i<0}\Bigr)
\eeq
and for $\la>0$,
\beq\label{mula}
\mu^{GW}_\la(\chi)=Z_\la^{-1}\exp(-\la H^{GW})\prod_{i\in\La_n}p_{X_i}
\eeq
where the partition function $Z_\la$ normalizes the probability. Then
$\mu^{GW}_\la$ converges in distribution to $\mu^{GW}$ as $\la\to+\infty$.
\end{proposition}
In other words, BGW configurations are the ground states of the BGW Hamiltonian
(\ref{HGW}). Note that by virtue of the boundary condition, we have $X_0\ge0$
with probability 1. 
\begin{proof}
\beq\label{11}
\prod_{i\in\La_n}\Bigl(1_{X_{a(i)}\ge r(i)}1_{X_i\ge0}+1_{X_{a(i)}<r(i)}1_{X_i<0}\Bigr)\ \in\{0,1\}
\eeq
because each factor is 0 or 1.
The inverse image of 1 by (\ref{11}) is the support of $\mu^{GW}$.
The set of nodes of $\om(\chi)$ is the set of
sites $i$ such that $X_i\ge0$. The indicator (\ref{11}) guarantees that
each node has a unique ancestor, which implies that $\om(\chi)$ is a tree.
Conversely the set of sites of
$\chi(\om)$ such that $X_i<0$ is $\La_n\setminus\om$. This proves (\ref{omchi}).
The Hamiltonian (\ref{HGW}) takes values
$H^{GW}(\chi)\in\{-|\La_n|,-(|\La_n|-1),\dots,0\}$.
It takes the value $-|\La_n|$ if and only if $\chi=\chi(\om)$ for some BGW tree
$\om$. There is a gap equal to one relative to the other states.
As $\la\to\infty$ the measure
concentrates on the ground states, where the Hamiltonian takes the value
$-|\La_n|$, which corresponds to (\ref{muGW}).
For curiosity, an example of $\chi$ such that $H^{GW}(\chi)=0$ is
\beq
X_i=\left\{
\begin{matrix}
-1&|i|\ {\rm even}\cr
K&|i|\ {\rm odd}
\end{matrix}
\right.
\eeq
where $|i|$ is the generation of node $i$.
\end{proof}

Now (\ref{phi1}) can be extended to
\beq\label{phi2}
\{-1,0,\dots,K\}\times\{-1,0,\dots,K\}\,\ni(X,Y)\longmapsto\,\vphi(X,Y)\,\in\R
\eeq
with $\vphi(X,-1)=\vphi(-1,X)=0\ \forall X$.
Given this isomorphism, from now on we'll use freely $\bra\cdot\ket_{GW}$
and $\bra\cdot\ket_n^x$ based on either representation of BGW trees, and
$\bra\cdot\ket_{\la,GW}$ and $\bra\cdot\ket_{\la,n}^x$ based on (\ref{mula}).
The interaction Hamiltonian will be denoted $H^x(\om)$ or $H^x(\chi)$ according
to the context.

For illustration, for any $i\in\La_n$, the probability that the random tree
$\om$ includes $i$, and the mean offspring of $i$ are respectively
\beq
\Pe^x_n(\om\ni i)=\Bigl\bra 1_{X_i\ge0}\Bigr\ket_n^x\qquad{\rm and}\qquad
\Bigl\bra X_i\,1_{X_i\ge0}\Bigr\ket_n^x
\eeq

The spin representation can also be viewed as a lattice gas representation with
$n_i=X_i+1$ the number of particles at $i\in\La_n$.

\section{Correlation inequalities}\label{correlation}
For definiteness we remain with a bounded number of offspring,
$X_i\leq K<\infty$, but correlation inequalities can be extended by
continuity to any offspring distribution, subject to existence of suitable
moments. Also the offspring distribution could depend upon the site $i\in\La_n$,
like a random field Ising model.
\subsection{Griffiths inequalities}
Following Ginibre \cite{G70}, let $\C_n$ denote the positive cone of
multinomials with 
non-negative coefficients in variables $f(X_i)$ where $|i|\le n$ and $f(\cdot)$ 
runs over non-negative non-decreasing functions on $\{-1,0,\dots,K\}$ with
$f(-1)=0$. 
\begin{lemma}\label{GGW}
Let $n\ge0$. Let $\om$ and $\om'$ be two independent
BGW trees with at most $n$ generations obeying the same probability law
$\Pe^{GW}$. Then for any family
$\{f_\al,\,i_\al\}_\al$ of non-negative non-decreasing functions $f_\al$ on
$\{-1,0,\dots,K\}$ with $f_\al(-1)=0$ and node labels $i_\al$ with $|i_\al|\le n$,
for any choices of $\pm$,
\beq\label{fal}
\sum_{\om,\om'}\Pe^{GW}(\om)\Pe^{GW}(\om')
\prod_\al\Bigl(f_\al( X_{i_\al})\pm f_\al( X'_{i_\al})\Bigr)\ge0
\eeq
Moreover for any $F,G\in\C_n$,
\beq
\bigl\bra F;G\bigr\ket_{GW}\ge0
\eeq
\end{lemma}
\begin{proof}
The second assertion is a straightforward consequence of the 
first, which we prove using (\ref{mula}), where we write
\beq\label{mula'}
\mu^{GW}_\la(\chi)\mu^{GW}_\la(\chi')\approx
\exp\bigl(-\la\bigl\{H^{GW}(\chi)+H^{GW}(\chi')\bigr\}\bigr)\prod_ip_{X_i}p_{X_i'}
\eeq
and then
\beqa
1_{X_{a(i)}\ge r(i)}1_{X_i\ge0}+1_{X_{a(i)}'\ge r(i)}1_{X_i'\ge0}=
\12\bigl(1_{X_{a(i)}\ge r(i)}+1_{X_{a(i)}'\ge r(i)}\bigr)
\bigl(1_{X_i\ge0}+1_{X_i'\ge0}\bigr)\cr
+\12\bigl(1_{X_{a(i)}\ge r(i)}-1_{X_{a(i)}'\ge r(i)}\bigr)
\bigl(1_{X_i\ge0}-1_{X_i'\ge0}\bigr)
\eeqa
\beqa\label{ql'}
1_{X_{a(i)}< r(i)}1_{X_i<0}+1_{X_{a(i)}'< r(i)}1_{X_i'<0}=
\12\bigl(1_{X_{a(i)}< r(i)}+1_{X_{a(i)}'< r(i)}\bigr)
\bigl(1_{X_i<0}+1_{X_i'<0}\bigr)\cr
+\12\bigl(1_{X_{a(i)}< r(i)}-1_{X_{a(i)}'< r(i)}\bigr)
\bigl(1_{X_i<0}-1_{X_i'<0}\bigr)
\eeqa
Expanding everything in (\ref{fal})(\ref{mula'}) yields a sum of terms
factorized over $i$, with each factor of the form
\beqa
\sum_{X,X'}p_Xp_{X'}\prod_k\bigl(1_{X\ge k}+1_{X'\ge k}\bigr)^{p_k}
\prod_{k'}\bigl(1_{X\ge k'}-1_{X'\ge k'}\bigr)^{q_{k'}}\cr
\prod_l\bigl(1_{X< l}+1_{X'< l}\bigr)^{p'_l}
\prod_{l'}\bigl(1_{X< l'}-1_{X'< l'}\bigr)^{q'_{l'}}\cr
\prod_{\al_i}\bigl(f_{\al_i}(X)+f_{\al_i}(X')\bigr)
\prod_{\al_i'}\bigl(f_{\al_i'}(X)-f_{\al_i'}(X')\bigr)
\eeqa
where the sums over $X,X'$ run over $\{-1,0,\dots,K\}$ and the products over
$k,k',l,l'$ run over $\{0,\dots,K\}$, while $p_k,q_{k'},p'_l,q'_{l'}$ are
collections of arbitrary fixed nonnegative integers. The indices $\al_i,\al_i'$
are for those $\al$ which fall on the given site.

The result is zero by symmetry if the number of factors
with - sign is odd. Otherwise, up to a factor 2, the summation can be restricted
to $X>X'$, where the summand has the sign $(-1)^{\sum_{l'}q'_{l'}}$. Indeed $1_{X<l}$
is a decreasing function while all others are increasing. The product over $i$
then yields a factor $(-1)^{\sum_i\sum_{l_i'}q'_{l_i'}}$, equal to $+1$, because all
$q'_{l'}$ factors come in pairs from (\ref{ql'}).
\end{proof}

\begin{theorem}\label{Gn} 
Let $n\ge0$ and $x\in\{1,\dots,K\}$. Assume (\ref{Pn})(\ref{Xin}) with
\beq\label{HCn}
-H^x(\chi)\in\C_n
\eeq
Then for any $F,G\in\C_n$,
\beq\label{gri}
\bigl\bra F;G\bigr\ket_n^x\ge0
\eeq
Moreover $\bra F\ket_n^x$ is non-decreasing in $n$ and in $\be$ and in
$x$ and in the coefficient of any term in $-H^x(\cdot)$ as an element of $\C_n$.
In particular $\forall\ 0\le m\le n$, the mean total offspring in generation $m$
obeys
\beq\label{gri2}
\bar k^{m+1}\le\Bigl\bra\sum_{|i|=m}X_i\,1_{X_i\ge0}\Bigr\ket_n^x\le K^{m+1}
\eeq
and is increasing in $n$ and has a limit as $n\to\infty$.
\end{theorem}
\begin{proof}
Inequality (\ref{gri}) is a standard consequence of Lemma \ref{GGW} \cite{G70}.
The first inequality in (\ref{gri2}) is comparison with $\be=0$, it follows
from (\ref{gri}). The second is trivial.
Monotonicity in $n$ is obtained as follows: consider an observable supported in
$\{|i|\le m\}$. Let $n'>n$. Let $\be=\be_{|j|}$ depend upon the generation of the
link $(j,a(j))$. For $0\le s\le 1$ define
$$
\bra\cdot\ket_{n,n',s}\ :\qquad \be_{|j|}=s\be \hbox{ for } n<|j|<n'
$$
Then
$$
\bra\cdot\ket_{n,n',0}=\bra\cdot\ket_n\,,\qquad
\bra\cdot\ket_{n,n',1}=\bra\cdot\ket_{n'}
$$
Monotonicity in $n$ follows from monotonicity in $s$, which follows from
 (\ref{gri}).
\end{proof}
\begin{remark}
Unlike the Ising model, here there is no spin flip symmetry, hence the
restriction to positive $f_\al$, $F$ and $G$.
\end{remark}

\subsection{FKG inequalities}
In the original paper \cite{FKG71}, the authors write that one can ``extend
straightforwardly to more general lattice gases where one allows more than one
particle on each site''. Here we give a statement and a proof for our model,
following \cite{FV17}.

\begin{theorem}\label{FKG}(FKG inequality)
Let $n\ge0$ and $x\in\{1,\dots,K\}$. Assume (\ref{phi1})-(\ref{Xin})(\ref{phi2})
with $\forall\ X,Y,X',Y'\in\{-1,0,\dots,K\}$
\beq\label{phifkg}
\vphi(X,Y)+\vphi(X',Y')\ge
\vphi(X\we X',Y\we Y')+\vphi(X\vee X',Y\vee Y')
\eeq
Then for any non-decreasing functions $F,G:\Om_n\rightarrow\R$
\beq\label{fkg}
\bigl\bra F;G\bigr\ket_n^x\ge0
\eeq
\end{theorem}
\begin{remark}
Interaction (\ref{phi}) obeys (\ref{phifkg}), even without the non-negativity
assumption.
\end{remark}
\begin{lemma}\label{th2}
Let
\beq
A_N=\{-1,0,1,\dots,K\}^N
\eeq
and let $f_1,f_2,f_3,f_4:\,A_N\rightarrow\R_+$ be such that
\beq\label{f1234}
f_1(\chi)f_2(\chi')\le f_3(\chi\we\chi')f_4(\chi\vee\chi')\qquad\forall\chi,\chi'\in A_N
\eeq
Then, for any product measure $\mu=\otimes\mu_i$ on $A_N$,
\beq
\bra f_1\ket_\mu\bra f_2\ket_\mu\le\bra f_3\ket_\mu\bra f_4\ket_\mu
\eeq
\end{lemma}
\begin{proof}
We first prove the theorem using the lemma with $A_N=\Om_n$, and then prove the
lemma. Let
\beq
h(\chi)=\prod_{i\in\La_n}\Bigl(1_{X_{a(i)}\ge r(i)}1_{X_i\ge0}+1_{X_{a(i)}<r(i)}1_{X_i<0}\Bigr)
\eeq
\begin{align}
f_1(\chi)=F(\chi)h(\chi)e^{-\be H^x(\chi)},\qquad &f_2(\chi)=G(\chi)h(\chi)e^{-\be H^x(\chi)},\cr
f_3(\chi)=h(\chi)e^{-\be H^x(\chi)},\qquad &f_4(\chi)=F(\chi)G(\chi)h(\chi)e^{-\be H^x(\chi)}
\end{align}
Applying the lemma gives the theorem. We must check (\ref{f1234}), which breaks
down into quadruples of factors and for which it suffices that every quadruple
obey the inequality. The quadruple containing $F$ and $G$ obeys the inequality
by virtue of the monotonicity of $F$ and $G$. The quadruples from the 2-body
interaction, embedded in exponentials, obey the inequality by hypothesis
(\ref{phifkg}). The Galton Watson quadruples obey the inequality because
$\forall\ r,X,Y,X',Y'\in\{-1,0,1,\dots,K\}$,
\begin{multline}\label{ppFKG}
\Bigl(1_{X\ge r}1_{Y\ge0}+1_{X<r}1_{Y<0}\Bigr)
\Bigl(1_{X'\ge r}1_{Y'\ge0}+1_{X'<r}1_{Y'<0}\Bigr)
\ \le \cr  
\Bigl(1_{X\wedge X'\ge r}1_{Y\wedge Y'\ge0}+1_{X\wedge X'<r}1_{Y\wedge Y'<0}\Bigr)
\Bigl(1_{X\vee X'\ge r}1_{Y\vee Y'\ge0}+1_{X\vee X'<r}1_{Y\vee Y'<0}\Bigr)
\end{multline}

Let us now prove the lemma. Let
\beq
\chi=(\ti\chi,x),\quad\chi'=(\ti\chi',y),\quad\ti\chi,\ti\chi'\in A_{N-1},\quad
x,y\in\{-1,\dots,K\}
\eeq
\beq
\ti f_j(\ti\chi)=\bra f_j(\ti\chi,\cdot)\ket_{\mu_N}
=\sum_{x=-1}^Kf_j(\ti\chi,x)\mu_N(x),\quad j=1,2,3,4
\eeq
We claim that the $\ti f_j$'s obey the hypothesis of the lemma, with $N-1$ in
place of $N$. Iterating $N$ times then proves the lemma. Alternatively
one can reason by induction, assuming the lemma up to $N-1$ and applying it to
the $\ti f_j$'s. In both cases there remains to prove the claim.
We start from the left-hand-side of (\ref{f1234}).
\begin{multline}\label{f1f2}
\ti f_1(\ti\chi)\ti f_2(\ti\chi')=
\bra f_1(\ti\chi,x)f_2(\ti\chi',y)\ket_{\mu_N\otimes\mu_N}\cr
=\bra 1_{x=y} f_1(\ti\chi,x)f_2(\ti\chi',y)\ket_{\mu_N\otimes\mu_N}\cr
+\bra 1_{x<y}\bigl[f_1(\ti\chi,x)f_2(\ti\chi',y)+f_1(\ti\chi,y)f_2(\ti\chi',x)\bigr]\ket_{\mu_N\otimes\mu_N}
\end{multline}
Given any $\chi,\chi'$ and $x<y$ let
\beqa
a=f_1(\ti\chi,x)f_2(\ti\chi',y),\qquad b=f_1(\ti\chi,y)f_2(\ti\chi',x)\cr
c=f_3(\ti\chi\wedge\ti\chi',x)f_4(\ti\chi\vee\ti\chi',y),\qquad
d=f_3(\ti\chi\wedge\ti\chi',y)f_4(\ti\chi\vee\ti\chi',x)
\eeqa
By hypothesis $a\le c$ and $b\le c$. Moreover
\beqa
ab=f_1(\ti\chi,x)f_2(\ti\chi',x)f_1(\ti\chi,y)f_2(\ti\chi',y)\le\cr
\le
f_3(\ti\chi\wedge\ti\chi',x)f_4(\ti\chi\vee\ti\chi',x)
f_3(\ti\chi\wedge\ti\chi',y)f_4(\ti\chi\vee\ti\chi',y)=cd
\eeqa
And $a,b\le c$ with $ab\le cd$ imply $a+b\le c+d$. 
The first term in (\ref{f1f2}) is bounded as
\beq
\bra 1_{x=y} f_1(\ti\chi,x)f_2(\ti\chi',y)\ket_{\mu_N\otimes\mu_N}
\le
\bra 1_{x=y} f_3(\ti\chi\wedge\ti\chi',x)f_4(\ti\chi\vee\ti\chi',y)\ket_{\mu_N\otimes\mu_N}
\eeq
The second term is bounded using $a+b\le c+d$:
\begin{multline}
\bra 1_{x<y}\bigl[f_1(\ti\chi,x)f_2(\ti\chi',y)+f_1(\ti\chi,y)f_2(\ti\chi',x)\bigr]\ket_{\mu_N\otimes\mu_N}\le\cr
\le
\bra 1_{x<y}\bigl[f_3(\ti\chi\wedge\ti\chi',x)f_4(\ti\chi\vee\ti\chi',y)
+f_3(\ti\chi\wedge\ti\chi',y)f_4(\ti\chi\vee\ti\chi',x)
\bigr]\ket_{\mu_N\otimes\mu_N}
\end{multline}
The proof of the lemma is now easily completed.
\end{proof}


\section{Recursion for generating functions}\label{recursion}
Let ${    N_n}$ be the number of {    external nodes} of an $n$-generation tree $\om_n$:
\beqa
{    N_n}(\om_n)&=&\sum_{i\in\om_n,\,|i|=n}X_{i_1\dots i_n}
\eeqa
Let ${    N_n}=L_n+Q_n$ where $L_n=L_n(\om_n)$ denotes the number of {    external nodes} 
whose parent has one {    offspring}, and $Q_n=Q_n(\om_n)$ denotes the 
number of {    external nodes} whose parent has two or more {    offspring}. For $u,v>0$ let
\beq\label{Xinuv}
\Xi^x_n(u,v)=\sum_{\om_n}\Pe^{GW}(\om_n)e^{-\be H^x(\om_n)}u^{L_n}v^{Q_n}
\eeq
with $H^x(\om_n)$ as (\ref{H})(\ref{22}), implying
$$
\Xi^x_n(u,v)=\Xi^2_n(u,v)\ \forall\ x\ge2
$$
The partition function
(\ref{Xin}) is $\Xi^x_n(1,1)$. Assume
\beq
e^{-\be \vphi(X,Y)}=e^\be=b>1\quad\forall\ X,Y\ge2\ ;\qquad
e^{-\be \vphi(X,Y)}=1\quad{\rm whenever}\ X\,{\rm or}\, Y\le1
\eeq
Then
\beqa\label{Xi01}
\left(\bmat\Xi^1_0(u,v)\\\Xi^2_0(u,v)\emat\right)
=\left(\bmat p_0+p_1u+p_2v^2+\dots+p_Kv^K\\p_0+p_1u+b\,p_2v^2+\dots+b\,p_Kv^K\emat\right)
\eeqa
\begin{theorem}
Let
\beqa\label{Fuv}
F(u,v)=\left(\begin{matrix}
p_0+p_1u+p_2v^2+\dots+p_Kv^K\\
p_0+p_1u+b\,p_2v^2+\dots+b\,p_Kv^K
\end{matrix}\right)
\eeqa
mapping $[1,\infty)\times[1,\infty)$ into itself and more precisely into
$\{1\leq u<v\}$. Then for $n\ge0$
\beq\label{XinFn}
\Xi^x_{n+1}(u,v)=\Xi^x_n(F(u,v))=\dots=\Xi^x_0\Bigl(F^{(n+1)}(u,v)
\Bigr)
\eeq
(Recursion from the {    external nodes}), and
\begin{align}
\Xi^1_{n+1}(u,v)&=p_0+p_1\Xi^1_n(u,v)+p_2\Bigl(\Xi^2_n(u,v)\Bigr)^2
+\dots+p_K\Bigl(\Xi^2_n(u,v)\Bigr)^K\cr
\Xi^2_{n+1}(u,v)&=p_0+p_1\Xi^1_n(u,v)+b\,p_2\Bigl(\Xi^2_n(u,v)\Bigr)^2
+\dots+b\,p_K\Bigl(\Xi^2_n(u,v)\Bigr)^K\cr
&\hskip 5cm\hbox{(Recursion from the root)}\label{fromroot}
\end{align}
Denote $(u_n(u,v),v_n(u,v))=F^{(n)}(u,v)$. Then
\beq\label{Xiuv}
\Xi^1_n(u,v)=u_{n+1}(u,v),\qquad \Xi^2_n(u,v)=v_{n+1}(u,v)
\eeq
When the arguments are not given, we implicitly assume $u=v=1$, so that
$(u_0,v_0)=(1,1)$ and
\beq
u_{n+1}=u_{n+1}(1,1)=\Xi_n^1(1,1)\,,\qquad v_{n+1}=v_{n+1}(1,1)=\Xi_n^2(1,1) 
\eeq
\beq\label{defF}
\left(\bmat u_{n+1}\cr v_{n+1}\cr\emat\right)=\left(\bmat p_0+p_1u_n+p_2v_n^2
+\dots+p_Kv_n^K\cr
p_0+p_1u_n+b\,p_2v_n^2+\dots+b\,p_Kv_n^K\cr\emat\right)
\eeq
\end{theorem}
\begin{proof}
Let $\om_n\subset\om_{n+1}$, and let $i\in {    N_n}$ mean $i$ an {    external node} of
$\om_n$, and similarly $i\in L_n$ and  $i\in Q_n$. Then
\beq
H^x(\om_{n+1})=H^x(\om_n)+\sum_{i\in {    N_n}}\vphi(X_{a(i)},X_i)
\eeq
\beqa\label{ruv}
\Xi^x_{n+1}(u,v)=&\sum_{\om_n}\Pe^{GW}(\om_n)e^{-\be H^x(\om_n)}
\hskip-0.17cm\prod_{i\in L_n}(p_0+p_1u+p_2v^2+\dots+p_Kv^K).\cr
&\hskip4cm.\prod_{i\in Q_n}(p_0+p_1u+b\,p_2v^2+\dots+b\,p_Kv^K)\cr
&\hskip-2.5cm=\sum_{\om_n}\Pe^{GW}(\om_n)e^{-\be H^x(\om_n)}
(p_0+p_1u+p_2v^2+\dots+p_Kv^K)^{L_n}.\cr
&\hskip4cm.(p_0+p_1u+b\,p_2v^2+\dots+b\,p_Kv^K)^{Q_n}\cr
=&\Xi^x_n\bigl(p_0+p_1u+p_2v^2+\dots+p_Kv^K,\,p_0+p_1u+b\,p_2v^2+\dots+b\,p_Kv^K\bigr)
\eeqa
proving (\ref{XinFn}). Relation (\ref{fromroot}) is straightforward, and
(\ref{Xiuv}) follows from (\ref{XinFn}) and (\ref{Xi01}).
\end{proof}
The recursion relation (\ref{defF}) means that the original statistical
mechanics problem has been reduced to a discrete time dynamical system in
$\R^2$, for which more efficient mathematical and numerical tools are available.
While the Monte Carlo simulation is limited to trees with $n$ about a hundred
generations, with the usual statistical errors, one can easily run
``exactly'' the dynamical system up to millions of generations, as shown on
Fig. \ref{Zn.png}.
It is also clear from the proof that for $b\ne1$, there is not a recursion
relation involving the diagonal $u=v$ alone.

We have lost the Markov chain of the noninteracting BGW model with $n$ as time.
Nevertheless we have gained a dynamical system with $n$ as time.

We shall also use the Fréchet derivative of the map $F$ defined in (\ref{Fuv}):
\beq\label{DF}
DF(u,v)=\left(\bmat 
p_1&2p_2v+\dots+Kp_Kv^{K-1}\cr
p_1&2b\,p_2v+\dots+Kb\,p_Kv^{K-1}
\emat\right)
\eeq

\section{Number of {    external nodes}}\label{leaves}
From (\ref{Xinuv}),
$$
\bra L_n\ket^x={\p\log\Xi^x_n(u,v)\over\p u}\Big|_{u=v=1},\qquad
\bra Q_n\ket^x={\p\log\Xi^x_n(u,v)\over\p v}\Big|_{u=v=1}
$$
\beq\label{Ln1}
\bra L_n\ket^1={\p\log u_{n+1}(u,v)\over\p u}\Big|_{u=v=1}
={1\over u_{n+1}}{\p u_{n+1}\over\p u}\Big|_{u=v=1}
\eeq
\beq
\bra L_n\ket^2={\p\log v_{n+1}(u,v)\over\p u}\Big|_{u=v=1}
={1\over v_{n+1}}{\p v_{n+1}\over\p u}\Big|_{u=v=1}
\eeq
\beq
\bra Q_n\ket^1={\p\log u_{n+1}(u,v)\over\p v}\Big|_{u=v=1}
={1\over u_{n+1}}{\p u_{n+1}\over\p v}\Big|_{u=v=1}
\eeq
\beq\label{Qn2}
\bra Q_n\ket^2={\p\log v_{n+1}(u,v)\over\p v}\Big|_{u=v=1}
={1\over v_{n+1}}{\p v_{n+1}\over\p v}\Big|_{u=v=1}
\eeq
$u_n,v_n$ and $\p u_n/\p u,\p u_n/\p v,\p v_n/\p u,\p v_n/\p v$ at $u=v=1$
can be computed by induction using
\beq\label{dudvnn}
\left(\bmat u_{n+1}\cr v_{n+1}\cr \p u_{n+1}/\p u\cr \p v_{n+1}/\p u\cr \p u_{n+1}/\p v\cr \p v_{n+1}/\p v\cr\emat\right)_{u=v=1}
=\left(\bmat p_0+p_1u_n+p_2v_n^2+\dots+p_Kv_n^K\cr
p_0+p_1u_n+b\,p_2v_n^2+\dots+b\,p_Kv_n^K\cr
p_1\p u_n/\p u+2p_2v_n\p v_n/\p u+\dots+Kp_Kv_n^{K-1}\p v_n/\p u\cr
p_1\p u_n/\p u+2b\,p_2v_n\p v_n/\p u+\dots+Kb\,p_Kv_n^{K-1}\p v_n/\p u\cr
p_1\p u_n/\p v+2p_2v_n\p v_n/\p v+\dots+Kp_Kv_n^{K-1}\p v_n/\p v\cr
p_1\p u_n/\p v+2b\,p_2v_n\p v_n/\p v+\dots+Kb\,p_Kv_n^{K-1}\p v_n/\p v
\emat\right)_{u=v=1}
\eeq
or, using (\ref{DF}),
\beq
\left(\bmat \p u_{n+1}/\p u\cr \p v_{n+1}/\p u\emat\right)_{u=v=1}
=DF(u_n,v_n)\left(\bmat \p u_n/\p u\cr \p v_n/\p u\emat\right)_{u=v=1}
\eeq
\beq
\left(\bmat \p u_{n+1}/\p v\cr \p v_{n+1}/\p v\emat\right)_{u=v=1}
=DF(u_n,v_n)\left(\bmat \p u_n/\p v\cr \p v_n/\p v\emat\right)_{u=v=1}
\eeq
It is worth noting that $u_n$ and $v_n$ are mutually coupled but independent
of $\p u_n/\p u$, $\p v_n/\p u$, $\p u_n/\p v$, $\p v_n/\p v$ at $u=v=1$, while
the latter depend upon $u_n,v_n$. The recursion starts with
\beq
\left(\bmat u_0\cr v_0\cr \p u_0/\p u\cr \p v_0/\p u\cr \p u_0/\p v\cr \p v_0/\p v\cr\emat\right)_{u=v=1}
=\left(\bmat 1\cr 1\cr 1\cr 0\cr 0\cr 1\emat\right),\qquad
\left(\bmat u_1\cr v_1\cr \p u_1/\p u\cr \p v_1/\p u\cr \p u_1/\p v\cr \p v_1/\p v\cr\emat\right)_{u=v=1}
=\left(\bmat 1\cr p_0+p_1+b\,p_2 +\dots+b\,p_K\cr p_1\cr p_1\cr 2p_2+\dots+Kp_K\cr 2b\,p_2+\dots+Kb\,p_K\emat\right)
\eeq
Using (\ref{Ln1})-(\ref{Qn2}) yields 
\beq
\bra L_0\ket^1=p_1,\qquad \bra L_0\ket^2={p_1\over p_0+p_1+b\,p_2+\dots+b\,p_K}
\eeq
\beq
\bra Q_0\ket^1=2p_2+\dots+Kp_K,\qquad \bra Q_0\ket^2={2b\,p_2+\dots+Kb\,p_K\over p_0+p_1+b\,p_2+\dots+b\,p_K}
\eeq
We have checked $\bra L_n\ket^x$ and $\bra Q_n\ket^x$ for $n=0,1,2$ and $x=1,2$
 and $K=2,3$ as above against the Monte Carlo algorithm.
The mean total number of {    external nodes} in a tree with at most $n$ generations,
with boundary condition 1 or 2, is given by
\beq\label{Zn12}
\bra {    N_n}\ket^1
={1\over u_{n+1}}\Bigl({\p u_{n+1}\over\p u}+{\p u_{n+1}\over\p v}\Bigr)_{u=v=1}\ ;\qquad
\bra {    N_n}\ket^2
={1\over v_{n+1}}\Bigl({\p v_{n+1}\over\p u}+{\p v_{n+1}\over\p v}\Bigr)_{u=v=1}
\eeq
\begin{figure}
\begin{center}
\resizebox{9cm}{!}{\includegraphics{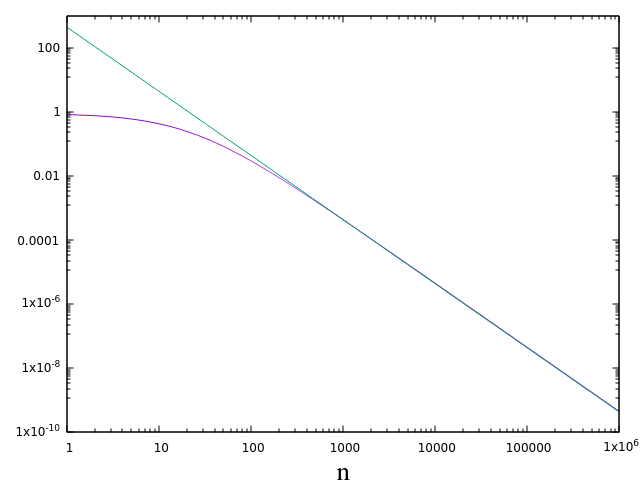}}
\caption{Mean number of {    external nodes} $\bra {    N_n}\ket^1$ and $n^{-2}$ fit for $K=2$,
$p_0=0.4$, $p_1=p_2=0.3$ and $\be=\be_c$ given by equality in (\ref{bec}).}
\label{Zn.png}
\end{center}
\end{figure}
Equation (\ref{Zn12}) with the recursion (\ref{dudvnn}) was used to generate
Fig. \ref{Zn.png}, further discussed in Section \ref{fixed}.

The variances can also be computed by induction, at the expense of four more
dimensions to accommodate
$\p^2 u_n/\p u^2,\p^2 v_n/\p u^2,\p^2 u_n/\p v^2,\p^2 v_n/\p v^2$. We give more
detail below for the variance of the energy.


\section{Interaction energy}\label{energy}
Let us write (\ref{H})(\ref{Pn})(\ref{Xin})(\ref{22}) as
\beq\label{Pn2}
\Pe^x_n(\om)=(\Xi^x_n)^{-1}\Pe^{GW}(\om)b^{N_{2{   2}}}
\eeq 
with
\beq\label{Xin2}
\Xi^x_n=\sum_{|\om|\le n}\Pe^{GW}(\om)b^{N_{2{   2}}}
\eeq
where $N_{2{   2}}=N_{2{   2}}(\om)$ is the number of favoured links: $X_i\ge2$ and
$X_{a(i)}\ge2$. It depends upon the boundary condition $x=X_{a(0)}$. Then
\beq
\bra N_{2{   2}}\ket_n^x=(\Xi^x_n)^{-1}b\,{d\over db}\Xi^x_n
\eeq
or, using (\ref{Xiuv}),
\beq\label{N2uv}
\bra N_{2{   2}}\ket_n^1={b\over u_{n+1}}\,{du_{n+1}\over db}\,,\qquad
\bra N_{2{   2}}\ket_n^2={b\over v_{n+1}}\,{dv_{n+1}\over db}
\eeq
The relation to energy is given by
\beq
e^{-\be H^x(\om)}=b^{N_{2{   2}}}\ \Rightarrow\ H^x(\om)=-{\log b\over\be}N_{2{   2}}
\eeq
The quantities (\ref{N2uv}) can be computed by recursion:
\beq\label{uvb}
\left(\bmat u_{n+1}\cr v_{n+1}\cr du_{n+1}/db\cr dv_{n+1}/db\emat\right)
=\left(\bmat p_0+p_1u_n+p_2v_n^2+\dots+p_Kv_n^K\cr
p_0+p_1u_n+b\,p_2v_n^2+\dots+b\,p_Kv_n^K\cr
p_1{du_n\over db}+2p_2v_n{dv_n\over db}+\dots+Kp_Kv_n^{K-1}{dv_n\over db}\cr
p_1{du_n\over db}+2bp_2v_n{dv_n\over db}+...+Kbp_Kv_n^{K-1}{dv_n\over db}\cr
\hfill+p_2v_n^2+\dots+p_Kv_n^K\emat\right)
\eeq
starting from (1,1,0,0) at $n=0$. Or, using (\ref{DF}),
\beq
\left(\bmat du_{n+1}/db\cr dv_{n+1}/db\emat\right)
=DF(u_n,v_n)\left(\bmat du_n/db\cr dv_n/db\emat\right)
\eeq
In particular
\beq
\left(\bmat u_1\cr v_1\cr du_1/db\cr dv_1/db\emat\right)
=\left(\bmat 1\cr
p_0+p_1+b\,p_2+\dots+b\,p_K\cr
0\cr
p_2+\dots+p_K\emat\right)
\eeq
\beq
\bra N_{2{   2}}\ket_0^1=0,\quad\bra N_{2{   2}}\ket_0^2={b\,(p_2+\dots+p_K)\over p_0+p_1+b\,(p_2+\dots+p_K)}
\eeq
We have checked $\bra N_{2{   2}}\ket_n^x$ for $n=0,1,2$ and $x=1,2$ and $K=2,3$ as above against
the Monte Carlo simulation with the algorithm of Section \ref{mc}.

The algorithm can be extended to the computation of variances:
\beq
\bra (N_{2{   2}})^2\ket_n^x-(\bra N_{2{   2}}\ket_n^x)^2=(\Xi^x_n)^{-1}b^2\,{d^2\over db^2}\Xi^x_n
+\bra N_{2{   2}}\ket_n^x-(\bra N_{2{   2}}\ket_n^x)^2
\eeq
\beq
\bra (N_{2{   2}})^2\ket_n^1-(\bra N_{2{   2}}\ket_n^1)^2={b^2\over u_{n+1}}\,{d^2u_{n+1}\over db^2}
+\bra N_{2{   2}}\ket_n^1-(\bra N_{2{   2}}\ket_n^1)^2
\eeq
\beq
\bra (N_{2{   2}})^2\ket_n^2-(\bra N_{2{   2}}\ket_n^2)^2={b^2\over v_{n+1}}\,{d^2v_{n+1}\over db^2}
+\bra N_{2{   2}}\ket_n^2-(\bra N_{2{   2}}\ket_n^2)^2
\eeq
where the second derivatives are also obtained recursively, adding to (\ref{uvb}) the following two lines:
\begin{align}\label{uvb2}
{d^2u_{n+1}\over db^2}=&
p_1{d^2u_n\over db^2}+2p_2v_n{d^2v_n\over db^2}+2p_2\Bigl({dv_n\over db}\Bigr)^2
+\dots\cr
&+Kp_Kv_n^{K-1}{d^2v_n\over db^2}+K(K-1)p_Kv_n^{K-2}\Bigl({dv_n\over db}\Bigr)^2\cr
{d^2v_{n+1}\over db^2}=&
p_1{d^2u_n\over db^2}+b\,\Bigl[2p_2v_n{d^2v_n\over db^2}+2p_2\Bigl({dv_n\over db}\Bigr)^2
+\dots\cr
&+Kp_Kv_n^{K-1}{d^2v_n\over db^2}+K(K-1)p_Kv_n^{K-2}\Bigl({dv_n\over db}\Bigr)^2\Bigr]
\end{align}
and starting the map in $\R^6$ from $(1,1,0,0,0,0)$.

\section{Fixed point and phase diagram}\label{fixed}
Let $K=2$. A fixed point $(u,v)$ for (\ref{defF}) reads
\beq\label{Pfp}
u={p_0+p_2v^2\over p_0+p_2}\,,\qquad v^2\Bigl(p_2(b-1)+{p_2\over p_0+p_2}\Bigr)-v
+{p_0\over p_0+p_2}=0
\eeq
which has a real solution if and only if
$$
b-1\le{p_0+p_2\over4p_0p_2}-{1\over p_0+p_2}
$$
\beq\label{bec}
\beta\le\log\Bigl(1+{p_0+p_2\over4p_0p_2}-{1\over p_0+p_2}\Bigr)=\be_c(p_0,p_2)
\eeq
See Fig. \ref{diag}.
For $b\ge1$ the system $(u,v)$, started at $(1,1)$, may only go to a fixed point
in $\{v\ge u\ge1\}$. The first equation in (\ref{Pfp}) will give a suitable
$u$ if $v\ge1$ has been found. We can therefore restrict our attention to the
second equation which may be written as
\beq\label{Pv}
P(v)=v^2p_2(b-1)+{p_2\over p_0+p_2}(v-1)(v-{p_0\over p_2})=0
\eeq
If $p_0<p_2$, corresponding to a free supercritical BGW, $P(v)$ cannot vanish and there is no fixed point
for the interacting system, where $v>1$. This is consistent with Griffiths
inequalities, Theorem \ref{Gn}, implying that $b>1$ increases the mean number of
{    external nodes}, reinforcing supercriticality.
\begin{figure}
\begin{center}
\resizebox{9cm}{!}{\includegraphics{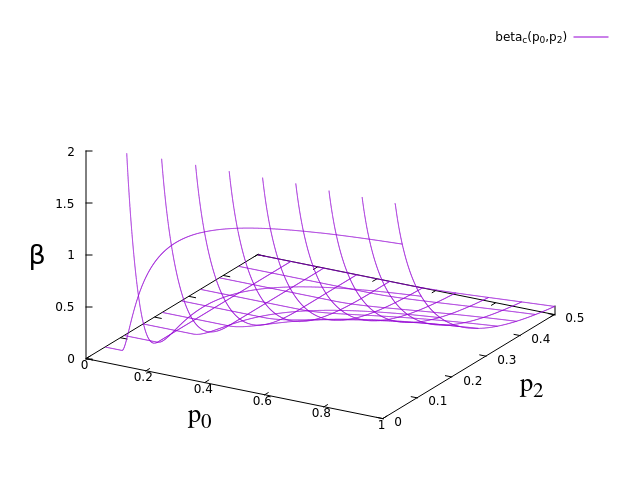}}
\caption{Critical surface given by equality in (\ref{bec}). The range of $(p_0, p_2)$ to consider is the triangle $p_2\ge0,\,p_2+p_0\le1,\,p_2-p_0\le0$. The system is expected to be subcritical below the critical surface, going to a fixed point as $n\nea\infty$, and supercritical above the critical surface, going to infinity as $n\nea\infty$.}
\label{diag}
\end{center}
\end{figure}
When $\be=\be_c(p_0,p_2)$, the fixed point associated with the double root is
\beq\label{uvfp}
v={2p_0\over p_0+p_2}\,;\qquad u={p_0\over p_0+p_2}+{4p_0^2p_2\over(p_0+p_2)^3}
\eeq

{    The phase diagram is linked to the behaviour of the free energy as $n\nea\infty$, depending upon $b$.} For boundary condition $x=$ 1 or 2, using (\ref{Lan})(\ref{Xiuv}), we find
\beq\label{psix}
\psi_n^x(b)=-{1\over|\La_n|}\log\Xi_n^x=-{1\over 2^{n+1}-1}\log\Xi_n^x
\eeq
\beq\label{psi}
\psi_n^1(b)=-{1\over 2^{n+1}-1}\log u_{n+1}\,,\qquad
\psi_n^2(b)=-{1\over 2^{n+1}-1}\log v_{n+1}
\eeq
Clearly, whenever there is convergence to a fixed point, the free energy density
(\ref{psix})(\ref{psi}) vanishes in the limit $n\nea\infty$.

When the underlying BGW model is critical or supercritical,
$\bar k=p_1+2p_2\ge 1$, the interacting model at $\be>0$ is supercritical and a
corresponding order parameter may be $\rho(b)$ from Theorem \ref{rho} below.
When the underlying BGW model is subcritical, $p_1+2p_2<1$,
there is $\be_c(p_0,p_2)$ such that the model is supercritical for $\be>\be_c$
and subcritical for $\be<\be_c$. Corresponding order parameters may be
\beq
\lim_{n\nea\infty}2^{-m}\Bigl\bra\sum_{|i|=m\atop i\in\om}X_i\Bigr\ket_n^x
\eeq
or
\beq
\lim_{n\nea\infty}2^{-(n-m)}\Bigl\bra\sum_{|i|=n-m\atop i\in\om}X_i\Bigr\ket_n^x
\eeq

An example for the mean number of {    external nodes} on the critical line is
shown on Fig.~\ref{Zn.png}. The behaviour $\bra {    N_n}\ket^1\sim n^{-2}$ as
$n\to\infty$ differs from the free
critical BGW where $\bra {    N_n}\ket_{GW}=\bar k^n=1$. 
{    A mathematical proof of this behaviour will be given in a forthcoming
paper \cite{CDHM22}.}
\section{Thermodynamic limit in the supercritical case.}\label{p_1p_2}
We give some results for the simplest interacting supercritical BGW model,
namely a model with $K=2$ and no extinction so that $p_1+p_2=1$, and interaction
(\ref{22}) as before.
\begin{theorem}\label{rho}
Let $p_1+p_2=1$ and $0<p_1<1$ and $b>1$. Recall $(u_n,v_n)=F^{(n)}(1,1)$ with
$F(\cdot)$ as (\ref{Fuv}). Then
\item{(i)} $u_n\to\infty,\,v_n\to\infty$, $v_n/u_n\to b$ as $n\to\infty$
\item{(ii)}
$(p_2bv_n)^{2^{-n}}$ strictly increases with $n$ for $n\ge0$.  
\item{(iii)}
$\rho(b)=\lim_{n\nea\infty}(v_n)^{2^{-n}}=\lim_{n\nea\infty}(u_n)^{2^{-n}}$ exists and is
non decreasing in b. 
\item{(iv)}
 $\rho(b)=\sup_n(p_2bv_n)^{2^{-n}}>1$
\item{(v)}
$\sqrt{p_2b(p_1+p_2b)}\le\rho(b)\le p_1+p_2b$. 
\item{(vi)}
$\psi(b)=\lim_{n\nea\infty}\psi_n^1(b)=\lim_{n\nea\infty}\psi_n^2(b)=-\log(\rho(b))$
exists and is non increasing in b. 
\end{theorem}
\noindent{\sl Remark:} The Theorem implies, in particular,
\begin{align}\label{llunvn}
\log\log v_n&=n\log 2+\log\log\rho+o(n)\cr
\log\log u_n&=n\log 2+\log\log\rho+o(n)
\end{align}
\begin{proof}
(i) Both $u_n$ and $v_n$ are bigger than they would be with the linear map
obtained by replacing $v^2$ by $v$ and omitting $p_0$ in (\ref{ruv}),
with which they would go to infinity exponentially as $n\to\infty$.
Since $v_n\ge u_n$ it follows that $v_n/u_n\to b$ as $n\to\infty$. 

\noindent (ii) We have $v_n>p_2bv_{n-1}^2$ for $n\ge1$, or 
$p_2bv_n>(p_2bv_{n-1})^2$, which gives (ii).

\noindent (iii-v)
Using $\{1\leq u\le v\}$ we have 
\beqa\label{vnupper}
v_n&\le&(p_1+p_2b)v_{n-1}^2\le(p_1+p_2b)((p_1+p_2b)v_{n-2}^2)^2\le\cdots\cr
&\le&(p_1+p_2b)^{1+2+4+\dots+2^{m-1}}v_{n-m}^{2^m}\le\cdots\cr
&\le&(p_1+p_2b)^{1+2+4+\dots+2^{n-2}}v_1^{2^{n-1}}\cr
&=&(p_1+p_2b)^{2^{n-1}-1}(p_1+p_2b)^{2^{n-1}}
=(p_1+p_2b)^{2^n-1}
\eeqa
which can be written as
\beq\label{iii}
(p_2bv_n)^{2^{-n}}\le(p_1+p_2b)\Bigl({p_2b\over p_1+p_2b}\Bigr)^{2^{-n}},
\qquad n\ge0
\eeq
The first claims in (iii-iv) and the upper bound in (v) follow from (i)(ii) and (\ref{iii}). The lower bound in (v) follows from
\beqa\label{vnlower}
v_n&>&p_2bv_{n-1}^2>p_2b(p_2bv_{n-2}^2)^2>\cdots>(p_2b)^{1+2+4+\dots+2^{m-1}}v_{n-m}^{2^m}\cr
&=&(p_2b)^{2^m-1}(v_{n-m})^{2^m}>\cdots>(p_2b)^{2^{n-1}-1}v_1^{2^{n-1}}\cr
&=&(p_2b)^{2^{n-1}-1}(p_1+p_2b)^{2^{n-1}}
=(p_2b)^{-1}\Bigl(p_2b(p_1+p_2b)\Bigr)^{2^{n-1}}
\eeqa
\noindent(iii) $\rho(b)$ is nondecreasing in $b$ because $v_n$ is a polynomial in $b$ with positive coefficients.

\noindent(iv) From the first part of (iv), $\rho(b)>1$ because $v_n\nea\infty$ as $n\nea\infty$.

\noindent(vi) follows from (iii) and (\ref{psi}). 
\end{proof}

\begin{figure}
\begin{center}
\resizebox{8cm}{!}{\includegraphics{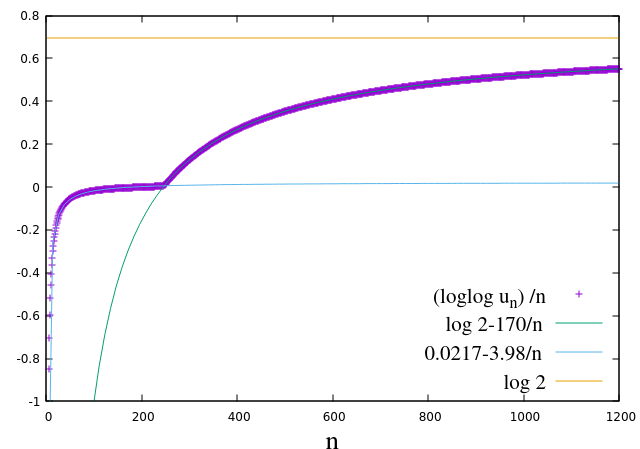}}
\caption{Numerical values of $\log(\log u_n)/n$ as function of $n$, for 
$p_1=0.9,\,b=1.1$, with asymptote at $\log 2$ and fits of the form $\ga-\del/n$ 
with ranges $[50:1000]$ (green) and $[10:50]$ (blue).}
\label{p109}
\end{center}
\end{figure}
\begin{figure}
\begin{center}
\resizebox{8cm}{!}{\includegraphics{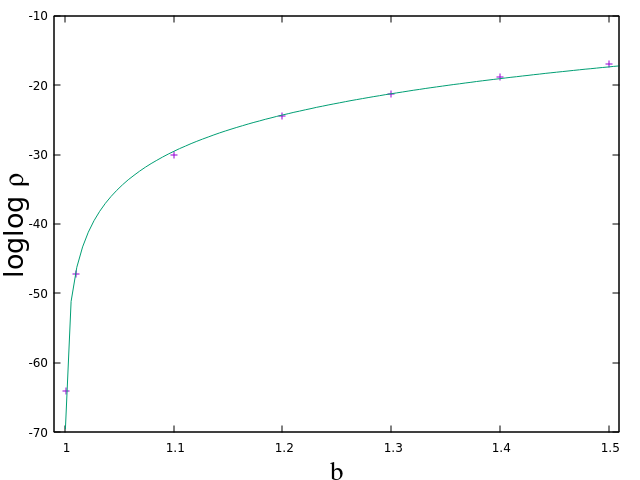}}
\caption{$\log\log\rho$ as function of $b$, for 
$p_1=0.9$, with fit $-12.163+7.5587\log(b-1)$.}
\label{vb}
\end{center}
\end{figure}
\begin{figure}
\begin{center}
\resizebox{9cm}{!}{\includegraphics{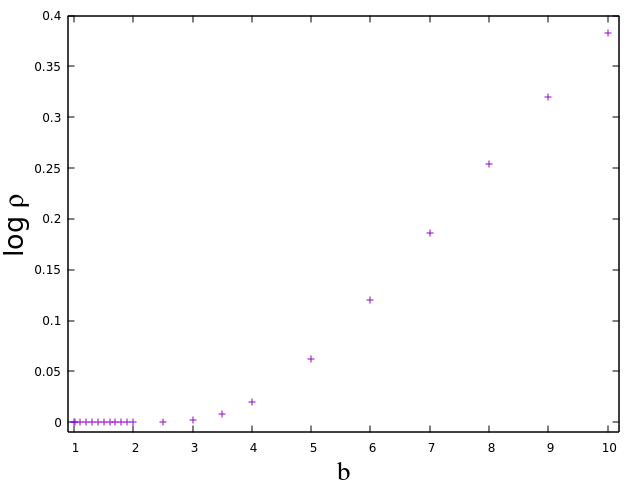}}
\caption{$\log\rho$ as function of $b$, for $p_1=0.9$.}
\label{expvb}
\end{center}
\end{figure}

If $p_2b(p_1+p_2b)<1$ the lower bound in (v) is useless.
We examine numerical values, see Fig. \ref{p109}.
For $p_1=0.9,\,b=1.1$ there seems to be a crossover: up to $n\sim 50$ there is
a good fit $\frac{1}{n}\log\log u_n\simeq 0.145-5.0/n$, while beyond  $n\sim 50$
there is a good fit $\frac{1}{n}\log\log u_n\simeq\log 2-30.0/n$ (up to 
$n=1000$, not shown). This is consistent with (\ref{llunvn}), with
$\log\log\rho\simeq-30.0$. The dependence of $\rho$ or $\psi$ upon $b$ for $p_1=0.9$ is sketched in Fig. \ref{vb} and Fig. \ref{expvb}.

\bigskip\noindent{\bf Data availability:}
 Data sharing not applicable to this article as no datasets were generated or analysed during the current study. All figures were drawn directly from the given two-dimensional dynamical system and simple use of the gnuplot ``fit''.

\bigskip\noindent{\bf Acknowledgements:} We would like to express our gratitude to Professor Ali Nesin as well as all the staff of Nesin Mathematics Village for their hospitality and the most pleasant environment for unhindered work they provided during our stay.
This work also benefitted at CY Cergy Paris University from the environment of
labex MME-DII
(Modèles Mathématiques et Economiques de la Dynamique, de l'Incertitude et
des Interactions), ANR11-LBX-0023-01. 

\end{document}